\documentclass[aps,10pt,a4paper,preprintnumbers,floatfix,superscriptaddress,pra,twocolumn,nofootinbib]{revtex4-1}
\pdfoutput=1
\usepackage[english]{babel}
\usepackage{amsmath,amssymb,graphicx,todonotes,bbm,xcolor,latexsym}
\newcommand{\mathd}{\mathrm{d}}
\newcommand{\textverbatim}[1]{{\ttfamily{#1}}}
\newcommand{\tmop}[1]{\ensuremath{\operatorname{#1}}}
\newenvironment{itemizeminus}{\begin{itemize} }{\end{itemize}}

\newenvironment{proof}{\noindent\textbf{Proof\ }}{\hspace*{\fill}$\Box$\medskip}
\newtheorem{definition}{Definition}
\newtheorem{proposition}{Proposition}

\begin{document}

\title{
  Universal bound on the cardinality of\\
  local hidden variables in networks
}

\author{Denis Rosset}
\email{physics@denisrosset.com}
\affiliation{Department of Physics, National Cheng Kung University, Tainan 701, Taiwan}
\affiliation{Group of Applied Physics,   Universit\'e de Gen\`eve, 1211 Gen\`eve, Switzerland}
\author{Nicolas Gisin}
\affiliation{Group of Applied Physics, Universit\'e de Gen\`eve, 1211 Gen\`eve, Switzerland}
\author{Elie Wolfe}
\affiliation{Perimeter Institute for Theoretical Physics, 31 Caroline St. N, Waterloo, Ontario, Canada, N2L 2Y5}

\date{June 12, 2017}

\begin{abstract}
  We present an algebraic description of the sets of local correlations in arbitrary networks, when the parties have finite inputs and outputs.
  We consider networks generalizing the usual Bell scenarios by the presence of multiple uncorrelated sources.
  We prove a finite upper bound on the cardinality of the value sets of the local hidden variables.
  Consequently, we find that the sets of local correlations are connected, closed and semialgebraic, and bounded by tight polynomial Bell-like inequalities.
\end{abstract}

\maketitle

Bell's theorem opened a new perspective for the study of quantum systems, as it predicted that quantum systems exhibit a wider range of correlations than
systems restricted to classical information.
The first results concerned two observers sharing a single resource modeled using a local hidden variable: the corresponding set of correlations is a polytope~{\cite{Pitowsky1991,Brunner2014}}.
This mathematical structure enables straightforward checking of the locality of correlations by linear programming~{\cite{Kaszlikowski2000}}, and the facets of the polytope provide ready-to-use linear inequalities.
In the multipartite version of Bell's locality~{\cite{Mermin1990,Brunner2014}}, the networks describe several observers sharing a single resource; there, the set of local correlations is still a polytope.
In all these studies, it is customary to identify the local hidden variable with a list of deterministic strategies implemented by the parties.
The (finite) number of those strategies provides an upper bound on the cardinality of the local hidden variable, that is the number of different values it has to take to reproduce all local correlations~{\cite{Donohue2015}}.

Later, the description of networks of uncorrelated sources led to local models containing independent local hidden variables~{\cite{Branciard2010,Branciard2012}}, extending the idea of ``local beables'' originated by John Bell~{\cite{Bell1964}}.
There, the correlation sets are no longer polytopes, being not even convex.
However, it is still possible, in some cases, to identify local hidden variables with deterministic strategies and to provide a bound on their cardinality~{\cite{Branciard2012}}.
The local sets of several networks have been characterized, at least partially, in the probability space~{\cite{Branciard2010,Branciard2012,Tavakoli2014,Rosset2015a,Tavakoli2016a,Tavakoli2016,Chaves2016a,Wolfe2016}} or in the entropy space~{\cite{Chaves2012,Chaves2014a,Henson2014}}, but there is no general method providing a list of inequalities in contrast to the case of Bell locality. However, the inflation technique~\cite{Wolfe2016} provides a hierarchy that converges~\cite{Navascues2017} to the local set.

As the description using marginal entropies loses information, we focus on the probability space in the present work.
As done usually for the characterization of local sets, we assume that inputs and outputs are taken from finite sets --- however, we do not assume this restriction on the local hidden variables {\em a priori}.
In any network, the characterization of the local correlations is tractable algebraically as long as one condition is satisfied: that {\em all} local hidden variables take a finite number of values.
When this condition holds, the set of local correlations is described by a system of polynomial inequalities~{\cite{Fritz2012}}, reminiscent of the linear Bell inequalities bounding the local set in Bell scenarios.
Correlations can also be tested for nonlocality by various algorithms.
Thus, an important open question is whether local hidden variables can be restricted to finite sets without loss of generality~{\cite{Branciard2012,Fritz2012}}.

Consider, as a motivating example, the triangle scenario (abbreviated $\Delta$) shown in Figure~\ref{Fig:ThreeScenarios}a, introduced independently in~\cite{Branciard2012} and~\cite{Steudel2015}.
The parties A, B, C share three independent local hidden variables with values $\alpha$, $\beta$, $\gamma$, according to the connections $A \leftarrow (\beta, \gamma)$, $B \leftarrow (\gamma, \alpha)$ and $C \leftarrow (\alpha, \beta)$ --- each party is connected to two variables, and outputs a single bit, written respectively $a, b, c = 0, 1$.
The joint probability distribution $P (a b c)$ describes the behavior of the network~\footnote{
  Writing $P (a b c)$ is an abuse of notation as it does not distinguish between the random variable $a$ itself and the {\em value} $a$ taken by the random variable in particular cases.
  When  necessary, we write $P (0)$ explicitly as $P (a = 0)$ or $P_a (0)$.
}.
We now consider a particular local behavior.
The local hidden variables are all uniformly distributed between 0 and 1.
Each party outputs the Boolean result of the comparison ``$\lambda_1 \geqslant \lambda_2$'', where $(\lambda_1, \lambda_2)$ corresponds to the pair of variables connected to it:
for example, $a = 1$ if and only if $\beta \geqslant \gamma$.
As the underlying $\Delta$-local model has a cyclic symmetry, the resulting correlations are symmetric under cyclic permutation of parties.
We compute easily $P (a b c) = 0$ if $a = b = c$ and $P (a b c) = 1 / 6$ otherwise.
Surprisingly, this behavior, which we write $\vec{P}_{\neq} = (0, 1, 1, 1, 1, 1, 1, 0) / 6$ by enumerating the indices $(a, b, c)$ in the lexicographic order, does not have a symmetric {\em finite} $\Delta$-local model, that is a symmetric model where $\alpha \in \Omega_{\alpha}$ with $\Omega_{\alpha}$ finite, as we now prove.

\begin{figure*}
  \includegraphics{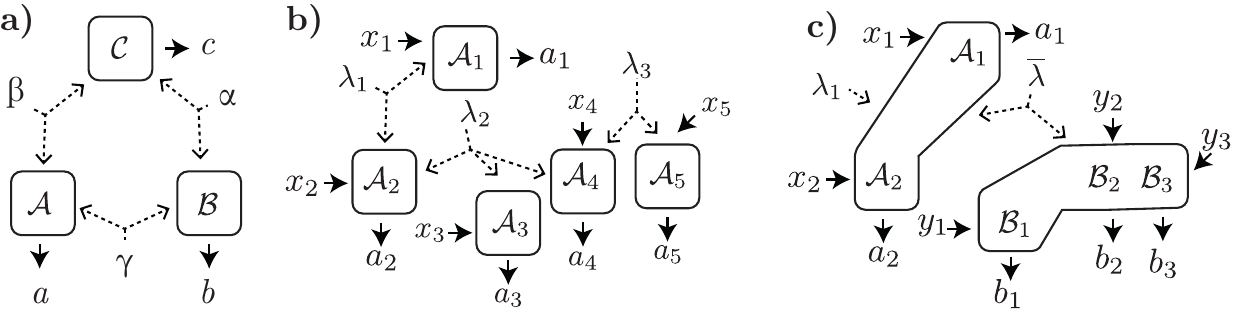}
  \caption{
    \label{Fig:ThreeScenarios}
    In a), the triangle network $\Delta$, where three parties share three bipartite local hidden variables $\alpha, \beta, \gamma$ and produce binary outputs, without receiving any input.
    In b), an example of complex network. In c), the renaming/grouping of parties used in Proposition~\ref{Prop:BetterBound} for the same network.
  }
\end{figure*}

Without loss of generality, a finite local model is written using local hidden variables $\alpha, \beta, \gamma \in \{ 1, \ldots, u \}$ with respective distributions $P_{\alpha}$, $P_{\beta}$, $P_{\gamma}$ and local response functions $P_{\text{A}}$, $P_{\text{B}}$, $P_{\text{C}}$:
\begin{multline}
  \label{Eq:FiniteModelTriangle}
  P(a b c) = \textstyle \sum_{\alpha, \beta, \gamma = 1}^u P_{\alpha} (\alpha) P_{\beta} (\beta) P_{\gamma} (\gamma) \\
  P_{\text{A}}(a | \beta \gamma) P_{\text{B}} (b | \gamma \alpha) P_{\text{C}} (c | \alpha \beta) .
\end{multline}

When the model is symmetric, we have $P_{\alpha} = P_{\beta} = P_{\gamma}$ and $P_{\text{A}} = P_{\text{B}} = P_{\text{C}}$.
To reproduce the correlations $\vec{P}_{\neq}$, we have $u = P_{\alpha} (\alpha = i) > 0$ and $v = P_{\text{A}} (a = j | \beta \gamma = ii) > 0$ for some values $i,j$.
Then $P (a b c = j j j) \geqslant u^3 v^3 > 0$ is a contradiction.
We could think of this as a hint that finite models do not exist for the correlations $\vec{P}_{\ne}$.
This intuition would be incorrect, as finite-valued models exist at the price of breaking symmetry (see Appendix~\ref{App:Asymmetric}).
Moreover, we show in this work that finite-valued models are actually universal, as long as the sets of inputs and outputs employed by the parties are themselves finite.

\paragraph*{Definitions. ---}
The study of nonlocality can be generalized to arbitrary networks of sources and parties.
Let us consider a network of $m$ parties sharing $n$ sources as in Figure~\ref{Fig:ThreeScenarios}b.
The parties are written $\mathcal{A}_1, \ldots, \mathcal{A}_m$, and have inputs $x_1, \ldots, x_m$ and outputs $a_1, \ldots, a_m$ taken from finite sets such that the observations are described by joint probability distribution:
\begin{equation}
  P (a_1 \ldots a_m | x_1 \ldots x_m) = P \left( \overline{a} \middle| \overline{x} \right)
\end{equation}
collecting $\overline{a} = (a_1, \ldots, a_m)$ and $\overline{x} = (x_1, \ldots, x_m)$.

By numbering the input and output values, we identify $x_i = 1, \ldots, X_i$ and $a_i = 1, \ldots, A_i$, and enumerate the coefficients of $P (a_1 \ldots a_m | x_1 \ldots x_m)$ in a vector $\vec{P} \in \mathbbm{R}^d$ where $d = X_1 \ldots X_m A_1 \ldots A_m$.
We describe the connections in the network by the incidence matrix $I \in \{0, 1 \}^{m \times n}$ where $I_{i j} = 1$ when the $i$-th party is connected to the $j$-th source. A network $\mathcal{N}$ is then described by the size of the input/output sets $\overline{X} = (X_1, \ldots, X_m)$, $\overline{A} = (A_1, \ldots, A_m)$ and the incidence matrix $I$.

In the network $\mathcal{N}$, we turn to the description of local models.
Each source $\mathcal{S}_j$ produces a local hidden variable $\lambda_j \in \Omega_j$, taken from the value set $\Omega_j$ with probability measure~\footnote{The usual notation used in the study of Bell locality uses a probability {\em density} $\rho_j  (\lambda_j)$ while the existence of a proper measure for $\lambda_j$ is implicitly assumed. The notation used here is equivalent and slightly shorter.} $\rho_j$.
Each party $\mathcal{A}_i$ receives an input $x_i$ along with the local hidden variables $\lambda_{[i]} = \{ \lambda_j | I_{i j} = 1 \}$, and processes them according to the response function $P_i (a_i | x_i \lambda_{[i]})$.
Then:
\begin{equation}
  \label{Eq:NetworkLocal}
  P \left( \overline{a} \middle| \overline{x} \right)
  = \prod_{j = 1}^n \int_{\Omega_j} \mathd \rho_j (\lambda_j) \prod_{i = 1}^m
  P_{\mathcal{A}_i} (a_i | x_i \lambda_{[i]}) .
\end{equation}
The local model $\mathcal{M}$ is fully described by the value sets $\Omega_j$, the probability densities $\rho_j$ and the response functions $P_i (a_i | x_i \lambda_{[i]})$.
The set of network-local (or $\mathcal{N}$-local) correlations $\mathcal{L}$ is the set of all $P \left( \overline{a} \middle| \overline{x} \right)$ reproduced by some model $\mathcal{M}$ according to~(\ref{Eq:NetworkLocal}).
Our main objective is to prove that all $\mathcal{N}$-local correlations can be produced by a simple model.

\begin{proposition}
  \label{Prop:GenericModel}
  Any $\vec{P} \in \mathcal{L} \subset \mathbb{R}^d$ can be reproduced with a generic model of the form:
  \begin{equation}
    \label{Eq:GenericModel}
    P \left( \overline{a} \middle| \overline{x} \right) =
    \prod_{j = 1}^n \sum_{\lambda_j = 1}^{d + 1} P_{\lambda_j}(\lambda_j) \prod_{i = 1}^m P_{\mathcal{A}_i} (a_i | x_i \lambda_{[i]}),
  \end{equation}
  where the local hidden variables $\lambda_j$ are integers in the set $\{ 1, \ldots, d + 1 \}$.
\end{proposition}

The proof will directly follow from Proposition~\ref{Prop:Cardinality} below, which bounds the cardinality of network-local models.

\paragraph*{Cardinality of network-local models. ---}
An estimation of the complexity of the model $\mathcal{M}$ is given by the cardinality of the input sets $\Omega_j$. When $\Omega_j$ is finite, we write $c_j = | \Omega_j |$ and otherwise $c_j = \infty$. The overall complexity is given by the tuple $\overline{c} = (c_1, \ldots, c_n)$.
We compare the power of given cardinalities by the componentwise partial order: $\overline{c} \leqslant \overline{c}'$ if $c_j \leqslant c'_j$ for all $j$.

We now look for an upper bound $\overline{c}_{\text{ub}}$ in any network such that all local $\vec{P} \in \mathcal{L}$ can be realized with models of cardinality $\overline{c} \leqslant \overline{c}_{\text{ub}}$.
Our bound is not minimal but scales linearly with the dimension of $\vec{P}$.
Note that the cardinalities $\overline{c}$ are not totally ordered and thus there could well be several minimal upper bounds $\overline{c}$ in a given network.

Our construction rests on the existence of individual bounds on the cardinality of each source.
We show that there exists a finite upper bound $u_1$ on the cardinality $c_1$ of the source $\mathcal{S}_1$: any model that reproduces a local behavior $\vec{P}$ with cardinality $c_1 > u_1$ can be modified into a model with $c'_1 \leqslant u_1$ without changing $c_2, \ldots, c_n$.
The same proposition holds for each source $\mathcal{S}_j$ with a corresponding upper bound $u_j$.
Thus, the cardinality $\overline{c}_{\text{ub}} = (u_1, \ldots, u_n)$ is sufficient to reproduce all local behaviors.

We first give a crude version of our construction.
Without loss of generality, we consider the cardinality of the source $\mathcal{S}_1$.

\begin{proposition}
  \label{Prop:Cardinality}
  In a given network $\mathcal{N}$, let $\mathcal{M}$ be a model for the local behavior $\vec{P} \in \mathbbm{R}^d$, with cardinality $\overline{c} = (c_1, c_2, \ldots, c_n)$.
  Then there is a model $\mathcal{M}'$ of cardinality $\overline{c}' = (d + 1, c_2, \ldots, c_n)$ that reproduces $\vec{P}$.
\end{proposition}

\begin{proof}
  Let us perform the following experiment with the model $\mathcal{M}$.
  We sample randomly a value $\lambda_1 = \mu \in \Omega_1$ from the source $\mathcal{S}_1$, and then replace $\mathcal{S}_1$ by a deterministic source that always outputs $\lambda_1 = \mu$.
  The behavior of the resulting model is:
  \begin{equation}
    P_{\mu} \left( \overline{a} \middle| \overline{x} \right) = \left[ \prod_{j = 2}^n \int_{\Omega_j} \mathd \rho_j (\lambda_j) \prod_{i = 1}^m P_{\mathcal{A}_i} (a_i | x_i \lambda_{[i]}) \right]_{\lambda_1 = \mu} .
  \end{equation}
  
  Now, the vector $\vec{P}_{\mu}$ is a random variable that depends on $\mu$ with average:
  \begin{equation}
    \langle \vec{P}_{\mu} \rangle = \int_{\Omega_1} \mathd \rho_1 (\mu) \vec{P}_{\mu},
  \end{equation}
  and this average is equal to the original behavior $\vec{P}$.
  Let $U = \{  \vec{P}_{\mu} | \mu \in \Omega_1 \} \subset \mathbbm{R}^d$ be the set of possible values of $\vec{P}_{\mu}$.
  By construction, $\langle \vec{P}_{\mu} \rangle$ is a convex mixture of points in $U$.
  Using Carath{\'e}odory's theorem (see Appendix~\ref{App:Caratheodory}), we can write $\langle \vec{P}_{\mu} \rangle$ as a convex combination of at most $d + 1$ points of $U$ with weights $\{w_k\}$:
  \begin{equation}
    \langle \vec{P}_{\mu} \rangle = \sum_{k = 1}^{d + 1} w_k \vec{P}_{\mu_k}
  \end{equation}
  as each point of $U$ can be realized with a $\mu \in \Omega_1$ (not necessarily unique).
  Then:
  \begin{equation}
    \vec{P} = \sum_{k = 1}^{d + 1} w_k \left[ \prod_{j = 2}^n \int_{\Omega_j} \mathd \rho_j (\lambda_j) \prod_{i = 1}^m P_{\mathcal{A}_i} (a_i | x_i \lambda_{[i]}) \right]_{\lambda_1 = \mu_k} .
  \end{equation}
  Thus, we can replace $\rho_1$ by a probability distribution $\rho_1'$ on a discrete set $\Omega_1' = \{ \mu_1, \ldots, \mu_{d + 1} \}$ with weights $w_k$, and obtain a model $\mathcal{M}'$ that reproduces the behavior $\vec{P}$ with cardinality $c_1 \leqslant d + 1$, while other elements of $\mathcal{M}$ are left unchanged.
\end{proof}

The proof of Proposition~\ref{Prop:GenericModel} follows directly.
Once all sets $\Omega_j$ have been replaced by sets of finite cardinality, we simplify the model structure by replacing all local variables by integers $\lambda''_1, \ldots ., \lambda''_n \in \{ 1, \ldots, d + 1 \}$ indexing the elements in the finite sets $\Omega_1', \ldots, \Omega_n'$.
In the end, any $\mathcal{N}$-local behavior can be reproduced by the generic model~(\ref{Eq:GenericModel}).

This result simplifies the study of local models in arbitrary networks.
When all the involved sets of values are finite, the set of correlations is parameterized by the generic model~(\ref{Eq:GenericModel}), a polynomial system involving a finite number of equations and unknowns (as already noted by various authors~\cite{Geiger2001,Fritz2012}.
This mathematical structure enables the generalization of several concepts used in the study of Bell locality.
As detailed in Appendix~\ref{App:Semialgebraic}, the set of network-local correlations is a closed semialgebraic set bounded by system of polynomial inequalities of the form $f (\vec{P}) \geqslant 0$ (for Bell, it was a polytope bounded by a finite number of linear inequalities).
Thus, any nonlocal behavior $\vec{P}$ violates such an inequality by a nonzero amount, a fact that can be tested experimentally.
The membership problem (is $\vec{P}$ local?) is a polynomial feasibility problem, solved for example by sum-of-squares relaxations that provide the relevant inequality (for Bell, a linear program).

Our result also simplifies the machinery of proofs, as it removes the conceptual difficulties of continuous models such as non-empty sets of measure zero.
An example is given in Appendix~\ref{App:Elementary}, where we provide an elementary proof that the behavior $P_{a b c} (000) = P_{a b c} (111) = 1 / 2$ is non-$\Delta$-local.

Our construction (and its refinements below) rely on Carath{\'e}odory's theorem, and is not directly constructive.
However, as it provides an upper bound on all $c_j$, we can always find a successful realization of $\vec{P}$ using a model of bounded cardinality.
We provide an example of such an exhaustive search in Appendix~\ref{App:Asymmetric}.

\paragraph*{Better upper bounds on the cardinality. ---}
We refine our bound by observing two properties of the set $U$ used in the proof of Proposition~\ref{Prop:Cardinality}.
Firstly, the bound provided by Carath{\'e}odory's theorem depends on the affine dimension of $U$, which is always less than $d$.
Secondly, we replace Carath{\'e}odory's theorem by the variant due to Fenchel, and use the fact that $U$ can always be taken as connected.

Let $\mathcal{P}$ be the set of arbitrary nonsignaling probability distributions $P\left( \overline{a} \middle| \overline{x} \right)$ in a network $\mathcal{N}$.
Its affine dimension~\cite{Boyd2004} is given by:
\begin{equation}
  \tmop{affdim} (\mathcal{P}) = \prod_{i = 1}^m [X_i (A_i - 1) + 1] - 1,
\end{equation}
where $(A_i, X_i)$ is the number of (outputs, inputs) of the $i$-th party.
This affine dimension is made explicit, for example, by the Collins-Gisin parameterization of $\mathcal{P}$~\cite{Collins2004}.

As before, we considering the local hidden variable $\lambda_1$ without loss of generality.
We write $\overline{\lambda} = (\lambda_2, \ldots, \lambda_n)$ with corresponding probability measure $\overline{\rho}$ over a set $\overline{\Omega}$.
We rename the parties and local hidden variables as follows.
As drawn in Figure~\ref{Fig:ThreeScenarios}c, we write $\mathcal{A}_1, \ldots, \mathcal{A}_N$ the parties connected to the local hidden variable $\lambda_1 \in \Omega_1$, with inputs $x_1, \ldots, x_N$ and outputs $a_1, \ldots, a_N$, collected in $\overline{a} = (a_1, \ldots, a_N)$, $\overline{x} = (x_1, \ldots, x_N)$.
We write $\mathcal{B}_1, \ldots, \mathcal{B}_{n - N}$ be the remaining parties, with inputs $y_1, \ldots, y_{n - N}$ and outputs $b_1, \ldots, b_{n - N}$, which we
also collected in $\overline{b}$ and $\overline{y}$.
The behavior $\vec{P}$ is written:
\begin{equation}
  \begin{array}{lll}
    &  & P (a_1 \ldots a_N b_1 \ldots b_{n - N} | x_1 \ldots x_{N_{}} y_1
    \ldots y_{n - N})\\
    & = & P \left( \overline{a}  \overline{b} \middle| \overline{x} 
    \overline{y} \right)\\
    & = & \int_{\Omega_1} \mathd \rho_1 (\lambda_1) \int_{\overline{\Omega}}
    \mathd \overline{\rho} \left( \overline{\lambda} \right)
    P_{\overline{\text{A}}} \left( \overline{a} \middle| \overline{x}
    \lambda_1 \overline{\lambda} \right) P_{\overline{\text{B}}} \left(
    \overline{b} \middle| \overline{y \lambda} \right),
  \end{array}
\end{equation}
where $P_{\overline{\text{A}}} \left( \overline{a} \middle| \overline{x} \lambda_1 \overline{\lambda} \right)$ collects the response functions of $A_1 \ldots A_N$ and $P_{\overline{\text{B}}} \left( \overline{b} \middle| \overline{y \lambda} \right)$ the response functions of $B_1 \ldots B_{n - N}$.

\begin{proposition}
  \label{Prop:BetterBound}
  Let $\mathcal{M}$ be a model for the local behavior $\vec{P} \in \mathbbm{R}^d$, of cardinality $\overline{c} = (c_1, c_2, \ldots, c_n)$.
  Then there is a model $\mathcal{M}'$ of cardinality $\overline{c}' = (u_1, c_2, \ldots, c_n)$ that reproduces $\vec{P}$ with $u_1 = \tmop{affdim} (\mathcal{P}_{\tmop{AB}}) - \tmop{affdim} \left( \mathcal{P}_{\text{B}} \right)$, where $\mathcal{P}_{\text{AB}}$ is the set of nonsignaling $P \left( \overline{a} \overline{b} \middle| \overline{x}  \overline{y} \right)$ and $\mathcal{P}_{\text{B}}$ the set of nonsignaling $P \left( \overline{b} \middle| \overline{y} \right)$.
\end{proposition}

\begin{proof}
  As before, for a fixed value $\lambda_1 = \mu$, we write:
  \begin{equation}
    P_{\mu} \left( \overline{a}  \overline{b} \middle| \overline{x} \overline{y} \right) =
    \left[ \int_{\overline{\Omega}} \mathd \overline{\rho}_{} \left( \overline{\lambda} \right)
      P_{\overline{\text{A}}} \left( \overline{a} \middle| \overline{x} \lambda_1 \overline{\lambda} \right)
      P_{\overline{\text{B}}} \left( \overline{b} \middle| \overline{y \lambda} \right) \right]_{\lambda_1 = \mu},
  \end{equation}
  with $U = \{ \vec{P}_{\mu} | \mu \in \Omega_1 \}$.
  Firstly, we remark that the marginal distribution $P_{\mu} \left( \overline{b} \middle| \overline{y} \right)$ does not depend on $\mu$:
  \begin{equation}
    \label{Eq:MarginalSame}
    P_{\mu} \left( \overline{b} \middle| \overline{x} \right) =
    \int_{\overline{\Omega}} \mathd \overline{\rho}_{} \left( \overline{\lambda} \right)
    P_{\overline{\text{B}}} \left( \overline{b} \middle| \overline{y \lambda} \right) =
    P \left( \overline{b} \middle| \overline{x} \right) .
  \end{equation}
  Secondly, the set $U$ can always be made connected by modifying the model as follows.
  We replace $\lambda_1$ by $\lambda_1' = (\lambda_1, \nu)$, where $\nu \in [0, 1]$ represents noise strength.
  We modify the response function of $\mathcal{A}_1$ such that:
  \begin{equation}
    P'_{\mathcal{A}_1} (a_1 | x_1 \lambda_{[1]} \nu) =
    \nu P_{\mathbbm{1}} (a_1 | x_1) + (1 - \nu) P_{A_1} (a_1 | x_1 \lambda_{[1]}),
  \end{equation}
  where $P_{\mathbbm{1}}$ is the uniformly random distribution, and so on for all $\mathcal{A}_i$.
  The new model still reproduces the behavior $\vec{P}$ provided we fix always $\nu = 0$.
  When we fix $\nu = 1$, we obtain:
  \begin{equation}
    P_{\nu = 1} \left( \overline{a}  \overline{b} \middle| \overline{x} \overline{y} \right) =
    P_{\mathbbm{1}} \left( \overline{a} \middle| \overline{x} \right) P \left( \overline{b} \middle| \overline{x} \right),
  \end{equation}
  which does not depend on $\mu$, and the path obtained by varying $\nu$ between 0 and 1 is continuous.
  Then, any point in $U$ can be brought to $\vec{P}_{\nu = 1}$, and $U$ is connected.
  
  An upper bound on the affine dimension of $U$ is given by the affine dimension of the set of nonsignaling $P \left( \overline{a}  \overline{b} \middle| \overline{x}  \overline{y} \right)$, after removing the degrees of freedom fixed by the constant marginal $P \left(  \overline{b} \middle| \overline{y} \right)$, see Eq.~\eqref{Eq:MarginalSame}.
  Thus:
  \begin{equation}
    \tmop{affdim} (U) \leqslant u_1 \equiv \tmop{affdim}(\mathcal{P}_{\tmop{AB}}) - \tmop{affdim} \left( \mathcal{P}_{\text{B}} \right) .
  \end{equation}
  As the set $U$ is connected, we apply Fenchel's variant of Carath{\'e}odory's theorem and the number of values in the convex decomposition is upper bounded by $u_1$.
\end{proof}

\paragraph*{Examples. ---}
In the $\Delta$ network presented in the introduction, we consider, without loss of generality, the cardinality of the variable $\alpha$.
The affine dimension of $P (a b c)$ is $7$ while the affine dimension of $P (a)$ is 1; thus $u_1 = 6$ and any $\Delta$-local distribution can be reproduced with $\Omega_{\alpha}, \Omega_{\beta}, \Omega_{\gamma}$ containing at most 6 values.

The bilocal scenario considered considered in~{\cite{Branciard2012}} has three parties, $\mathcal{A}$, $\mathcal{B}$ and $\mathcal{C}$, with binary inputs and outputs $a, b, c, x, y, z = 0, 1$ connected by local hidden variables $\lambda_{\mathcal{A}\mathcal{B}}$ and $\lambda_{\mathcal{B}\mathcal{C}}$.
The affine dimension of $P (a b c | x y z)$ is 26, the affine dimension of $P (b c|y z)$ is 8; thus the cardinality of $\lambda_{\mathcal{A}\mathcal{B}}$ is upper bounded by $18$.
However, by enumerating the deterministic strategies corresponding to $\mathcal{A}$, the cardinality of $\lambda_{\mathcal{A}\mathcal{B}}$ is maximum $4$.

This shows that the upper bound presented in this paper is not always optimal.
Indeed, consider a Bell scenario where $\mathcal{A}$ and $\mathcal{B}$ have binary outputs $a, b = 0, 1$ but no input, and are connected to the variable $\lambda$.
The affine dimension of $P (a b)$ is $3$, and thus we obtain the upper bound $3$ on the cardinality of $\lambda$.
However, we can also identify $\lambda = a$ and write $P (a b) = P (a) P (b | a)$ to provide a model with cardinality 2.

\paragraph*{Extensions to quantum resources. ---}
Our result extends to resources other than local hidden variables.
Consider, for example, the bipartite scenario where a quantum source produces the state $\rho_{\text{AB}} = \sum_i p_i | \varphi_i \rangle \langle \varphi_i |$.
Then, we observe:
\begin{align}
  P(ab|xy) &= \tmop{tr} \left[ \rho_{\text{AB}}  \left( \Pi^{\text{A}}_{a | x} \otimes \Pi^{\text{B}}_{b | y} \right) \right] \nonumber \\
  &= \sum_i  p_i \left\langle \varphi_i \middle| \Pi^{\text{A}}_{a | x} \otimes \Pi^{\text{B}}_{b | y} \middle| \varphi_i \right\rangle .
\end{align}
By convexity, an argument similar to Proposition~\ref{Prop:Cardinality} selects at most $d + 1$ elements in the decomposition of $\rho_{\text{AB}}$, so that the {\em rank} of $\rho_{\text{AB}}$ can be reduced, while its {\em dimension} is left unchanged.
In contrast, the complexity of local hidden variables is measured by a single parameter: the cardinality of their value set.

\paragraph*{Conclusion. ---}

We proved that finite local hidden variables suffice to reproduce local correlations in any network.
In consequence, we showed that local sets of correlations are semialgebraic, tightly bounded by a finite number of polynomial inequalities.
We also outlined algebraic methods to check the nonlocality of given correlations.
Our construction rests on the dimension of the output space and does not exploit much of the network structure.
It is thus likely that the bound we propose can be improved.
For example, in the case of the triangle network $\Delta$, Proposition~\ref{Prop:BetterBound} proposes an upper bound of $6$, while the behavior $P_{\neq}$ has a local model using only bits.
We leave as an open question the closing of the gap between these bounds.

However, for some purposes, a tight upper bound does not matter as only the {\em existence} of a finite model simplifies the calculations.
Indeed, finite probability spaces are conceptually much simpler, as the machinery of probability measures/densities can replaced by discrete distributions.
Finally, we exhibited a generalization of our method to quantum states, and outlined the distinct notions of rank and dimension, which are conflated in the classical case.
Interestingly, our method can only bound the rank of a quantum state but not its dimension.

\paragraph*{Acknowledgements. ---}
We thank many colleagues for discussions in the last three years; among them Jean-Daniel Bancal, Nicolas Brunner, Rafael Chaves, Flavien Hirsch, Yeong-Cherng Liang, Marc-Olivier Renou and Rob Spekkens.
This work is supported by the Swiss National Science Foundation via the Mobility Fellowship P2GEP2\_162060 and the NCCR-QSIT; also by Perimeter Institute for Theoretical Physics. Research at Perimeter Institute is supported by the Government of Canada through the Department of Innovation, Science and Economic Development and by the Province of Ontario through the Ministry of Research and Innovation. We also acknowledge support by the Ministry of Education, Taiwan, R.O.C., through “Aiming for the Top University Project” granted to the National Cheng Kung University (NCKU), and by the Ministry of Science and Technology, Taiwan (Grants No. 104-2112-M-006-021-MY3). 

\

\bibliography{cardinality}

\

\appendix

\section{Variants of Carath{\'e}odory's theorem}
\label{App:Caratheodory}
Let $X \subset \mathbb{R}^d$ a bounded subset of $\mathbb{R}^d$ and $\rho$ a probability measure on (the Borel subsets of) $\mathbb{R}^d$.
We assume that $X$ has full affine dimension $a = d$.
Otherwise, when $a < d$, we use the fact that affine maps preserve convex decompositions, and study the image of $X$ under an injective affine map $f: \mathbb{R}^d \rightarrow \mathbb{R}^a$.
As noted in several steps of the proof below, we make liberal use of affine maps to bring points where convenient.

Let $\vec{x}^*$ be the center of mass of $\rho$:
\begin{equation}
  \vec{x}^* = \int_X \mathd \rho(\vec{x}) ~ \vec{x}.
\end{equation}
When $X$ is compact, the standard formulation~\cite{Barany2012} of Carath{\'e}odory's theorem states that there exists at most $d+1$ elements $\vec{x}_i \in X$ such that $\vec{x}^*$ is a convex mixture of those:
\begin{equation}
  \vec{x}^* = \sum_{i=1}^{d+1} w_i \vec{x}_i, \quad \sum_i w_i = 1, \quad \forall i, w_i \ge 0.
\end{equation}
Moreover, when $X$ is also connected, Fenchel's variant~\cite{Barany2012} of the theorem reduces the upper bound to $d$ instead of $d+1$.

When the set $X$ is not closed, we first reduce to the finite case before applying the theorem.
Without loss of generality, we shift $X$ such that $\vec{x}^*$ is the origin.
We rewrite $\vec{x} = (y, \vec{z})$ with $y \in \mathbb{R}$ and $\vec{z} \in \mathbb{R}^{d-1}$ and split $X$ into three sets: $X_+$, $X_0$ and $X_-$ depending on the sign of $y$ (see Figure~\ref{Fig:Caratheodory}).
\begin{figure}
  \includegraphics{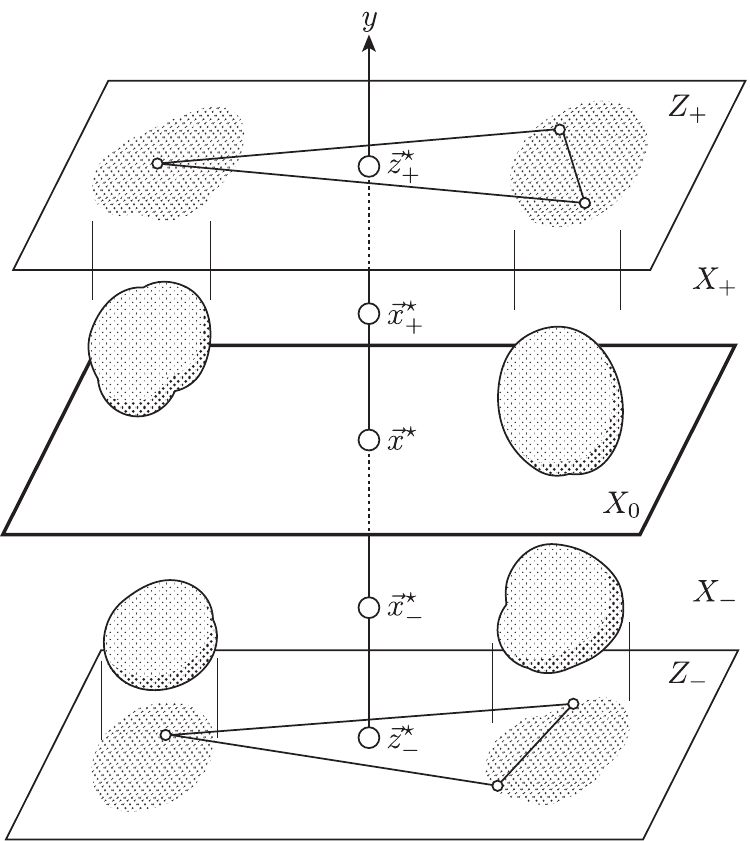}
  \caption{
    \label{Fig:Caratheodory}
    For $d=3$, the decomposition of $X$, represented by dotted blobs, into sets $X_\pm$ depending on the sign of $y$.
    The top and bottom planes represent the projection $Z_\pm$, with the triangles showing that $\vec{z}^\star_\pm$ lies in the convex hull of at most $d$ points.
  }
\end{figure}
We restrict and normalize $\rho$ to each of these three sets to obtain $\rho_+$, $\rho_0$ and $\rho_-$, and compute:
\begin{equation}
\vec{x}^*_0 = \int_{X_0} \mathd \rho_0(\vec{x}) ~ \vec{x}, \qquad  \vec{x}^*_{\pm} = \int_{X_\pm} \mathd \rho_\pm(\vec{x}) ~ \vec{x},
\end{equation}
such that:
\begin{equation}
  \vec{x}^* = w_+ \vec{x}^*_+ + w_0 \vec{x}^*_0 + w_- \vec{x}^*_-
\end{equation}
for nonnegative, normalized weights $w_0, w_\pm$.
When any of $X_\pm$ or $X_0$ is empty (or of measure zero), we simply remove it from this decomposition.
Note that $\vec{x}^*_\pm$, $\vec{x}^*_0$ are convex mixtures on their own and are not necessarily elements of the sets $X_\pm$, $X_0$.
However, $\vec{x}^*_\pm$, $\vec{x}^*_0$ can be obtained by a mixture of a finite number of points of those sets as we show below.

We proceed by induction.
When $d=1$, the solution is trivial when $X_0 \ne \emptyset$.
Otherwise, both $X_+$ and $X_-$ are not empty and we can replace $\vec{x}^*_\pm$ by any point in $X_\pm$, adjusting the weights $w_\pm$ as required.

For $d>1$, we first look at $\vec{x}^*_0$.
If $X_0$ is not empty, we find a finite decomposition for $\vec{x}^*_0$ by applying the proposition for $(d-1)$.
We then write $\vec{x}^*$ as a mixture of $\vec{x}^*_0$ and the point $\vec{x}^\star = (w_+ \vec{x}^*_+ + w_- \vec{x}^*_-) /(w_+ + w_-)$.
It remains to show that the point $\vec{x}^\star$ has a finite decomposition (and now, we identify $\vec{x}^\star = \vec{x}^*$ when $X_0$ has measure zero).
Without loss of generality, we assume again that $\vec{x}^\star = \vec{0}$.
We write:
\begin{equation}
  \vec{x}^\star = \begin{pmatrix} 0 \\ \vec{0} \end{pmatrix} = \frac{1}{w_+ + w_-} \left [ w_+ \begin{pmatrix} y^\star_+ \\ \vec{z}^\star_+ \end{pmatrix} + w_- \begin{pmatrix} y^\star_- \\ \vec{z}^\star_- \end{pmatrix} \right ]. 
\end{equation}

Using an affine transform, we can always ensure that $\vec{z}^\star_\pm = \vec{0}$ without changing the $y$ coordinates (shear mapping).
We first consider the projection of $X_+$ on the last $d-1$ coordinates:
\begin{equation}
  Z_+ = \{ \vec{z}_+ \text{ s.t. } \exists y_+, (y_+, \vec{z}_+) \in X_+ \}.
\end{equation}
By using the proposition for $d-1$, we find a convex decomposition of $\vec{z}^\star_+$ using a finite number of points:
\begin{equation}
  \vec{0} = \vec{z}_+^\star = \sum_i p_i \vec{z}_i, \quad \vec{z}_i \in Z_+, \quad p_i \ge 0, \quad \sum_i p_i = 1.
\end{equation}
By picking for each $\vec{z}_i$ a corresponding point $\vec{x}_i = (y_i, \vec{z}_i) \in X_+$, we obtain a convex decomposition:
\begin{equation}
  \vec{x}_+^\circ = \begin{pmatrix} y^\circ_+ \\ \vec{0} \end{pmatrix} = \sum_i p_i \begin{pmatrix} y_i \\ \vec{z}_i \end{pmatrix},
\end{equation}
where $y^\circ_+ > 0$ by construction.
The same argument for $X_-$ provides a convex decomposition of $\vec{x}_-^\circ = (y^\circ_-, \vec{0})$ with $y^\circ_- < 0$.
Thus, $\vec{x}^\star = \vec{0}$ can be written as a convex mixture of the two points $\vec{x}_\pm^\circ$, which can in turn be written as a combination of a finite number of points.

Note that in Proposition~\ref{Prop:BetterBound}, we use the fact that $X$ is closed {\em and} connected, which can always be achieved by first reducing to the finite case, then adding the connectivity using the trick mentioned in the proof.

\section{Asymmetric models for $P_{\neq}$}\label{App:Asymmetric}

We provide in this Appendix two asymmetric finite $\Delta$-local models for the behavior $\vec{P}_{\neq}$ presented in the introduction.
The first model has cardinality $\overline{c} = (| \Omega_{\alpha} |, | \Omega_{\beta} |, | \Omega_{\gamma} |) = (3, 2, 6)$, while the second model has $\overline{c} =
(2, 2, 2)$.

\paragraph*{By application of our proposition. ---}
We take the original model and apply the construction of Proposition~\ref{Prop:Cardinality} to each variable in turn, modifying their distributions such that the model always reproduces $\vec{P}_{\neq}$ after each step.
We start with $\beta$, and $|\Omega'_{\beta} | > 1$, otherwise $P_{\neq}$ would factorize as $P_{\neq} (a c) = P_{\neq} (a) P_{\neq} (c)$.
We try with $\beta_+, \beta_- \in \Omega_{\beta}$ and a weight $v \in [0, 1]$, such that $\beta_+$ is distributed with probability $v$ and $\beta_-$ with probability $1 - v$.

We obtain the solution $\beta_{\pm} = \left( 3 \pm \sqrt{3} / 6 \right)$, with $v = 1 / 2$.
We turn to $\alpha$, and first try to use three $\alpha_1, \alpha_2, \alpha_3 \in \Omega_{\alpha}$ distributed with the respective nonnegative weights $u_1 + u_2 + u_3 = 1$.
We obtain a model with $u_1 = u_3 = \left( 3 - \sqrt{3} \right) / 6$, $\alpha_1 = \left( 3 - \sqrt{3} \right) / 12$, $\alpha_2 = 1 / 2$ and $\alpha_3 = 1 - \alpha_1$.
We turn finally to $\gamma$, and observe that by the definition of the response functions ``$\lambda_1 \geqslant \lambda_2$?'', the most general set of values for $\gamma$ is given by $\Omega_{\gamma}' = \{ \gamma_1, \ldots, \gamma_6 \}$ where the $\gamma_k$ are contained in the six strict intervals between the seven numbers $\{ 0, \alpha_1, \beta_1, \alpha_2, \beta_2, \alpha_3, 1 \}$ and the exact values of $\gamma_k$ do not matter.
We consider a distribution of $\gamma_k$ with weights $w_k$, and obtain $w_1 = w_2 = w_5 = w_6 = \left( 3 - \sqrt{3} \right) / 12$ and $w_3 = w_4 = 1 / \left( 2 \sqrt{3} \right)$.
Our transformed model has $| \Omega_{\alpha}' | = 3$, $| \Omega_{\beta}' | = 2$ and $| \Omega_{\gamma}' | = 6$.

\paragraph*{Using a combinatorial approach. ---}
We found a minimal model for the correlations $\vec{P}_{\neq}$ with $\overline{c} = (2, 2, 2)$ using an hybrid combinatorial-algebraic search.
We start with $\alpha, \beta, \gamma = 0, 1$ and all coefficients $P_{\alpha} (\alpha), P_{\beta} (\beta), P_{\gamma} (\gamma)$ {\em strictly} between 0 and 1 (as any $c_i = 1$ is impossible).
We first solve the following relaxed problem.
Instead of computing the exact value of $P (a | \beta \gamma)$, we only tracked the possible values of $a$ for the four $(\beta, \gamma)$ pairs:
whether $P (a = 0 | \beta \gamma) = 0$, $P (a = 0 | \beta \gamma) = 1$ or $P(a = 0 | \beta \gamma) \in] 0, 1 [$; the same for $P (b | \gamma \alpha)$ and $P (c | \alpha \beta)$.
For example, it is impossible to have $P (a = 0 |\beta \gamma) \in] 0, 1 [$ for all $(\beta, \gamma)$ pairs, as it renders the condition $P (000) = 0$ impossible to satisfy.

In total, there are half a million ($3^{12}$) possibilities to check.
After removing the obvious symmetries (permutation of parties, bit flips of $\alpha, \beta, \gamma$), we are left with $4$ cases corresponding to polynomial
feasibility problems of degree 4 involving between 4 and 5 unknowns.
We checked those cases manually using a branch-and-bound solver (BMIBNB~{\cite{Lofberg2004}}) and found the following model.
The local hidden variables have distributions $P_{\alpha} (0) = P_{\beta} (0) = 1 / 3$ and $P_{\gamma} (0) = 1 / 4$.
The response functions of A and B are deterministic such that $a = \beta \gamma$ and $b = 1 \oplus \gamma \alpha$, while Charlie uses an additional uniformly random bit $\nu$ so that $c = \nu$ when $\alpha = \beta$, and $c = \alpha$ otherwise.

Without solving the combinatorial problem first, the algebraic problem involves monomials of degree 6 in 15 variables, and is thus out of reach of branch-and-bound or sum-of-squares relaxation methods.

\section{Example of an elementary proof using finite models}
\label{App:Elementary}

In the $\Delta$ network, consider the correlations $\vec{P}_= = (1 / 2, 0, 0, 0, 0, 0, 0, 1 / 2)$ where the outputs are always correlated ($a = b = c$).
These correlations were proven non-$\Delta$-local~{\cite{Fritz2012}} using entropic inequalities, or in~\cite{Wolfe2016} using the inflation technique.
We provide an elementary proof, based on the generality of finite $\Delta$-local models.

Using Proposition~\ref{Prop:GenericModel}, we assume that the local hidden variable sets $\Omega_{\alpha}$, $\Omega_{\beta}$ and $\Omega_{\gamma}$ are finite.
By embedding any local source of randomness in a local hidden variable, we can assume the local response functions to be deterministic~{\cite{Brunner2014}}.
Because $P_{\text{=}} (000) > 0$, there exists a triplet of values $(\alpha_0, \beta_0, \gamma_0)$ resulting in\footnote{That is $P_{\text{A}} (0 | \beta_0 \gamma_0) = P_{\text{B}} (0 |\gamma_0 \alpha_0) = P_{\text{C}} (0 | \alpha_0 \beta_0) = 1$.} $a = b = c = 0$, with the triplet probability $P_{\alpha \beta \gamma} (\alpha_0 \beta_0 \gamma_0) > 0$.
A similar argument exhibits $(\alpha_1, \beta_1, \gamma_1)$ such that $a = b = c = 1$ with $P_{\alpha \beta \gamma} (\alpha_1 \beta_1 \gamma_1) > 0$.
Now, consider $(\alpha, \beta, \gamma) = (\alpha_0, \beta_0, \gamma_1)$; surely, $P_{\alpha \beta \gamma} (\alpha_0 \beta_0 \gamma_1) > 0$.
Because $c = 0$, we must have $a = 0$ whenever $(\beta, \gamma) = (\beta_0, \gamma_1)$.
Now, when $(\alpha, \beta, \gamma) = (\alpha_1, \beta_0, \gamma_1)$ we have $b = 1$, and thus we must have $a = 1$ for $(\beta, \gamma) = (\beta_0, \gamma_1)$, which is a contradiction.

The proof does not hold when $\Omega_{\alpha}, \Omega_{\beta}, \Omega_{\gamma}$ are infinite: while the probability density $\rho_{\alpha} (\alpha_0) \rho_{\beta} (\beta_0) \rho_{\gamma} (\gamma_1)$ could be nonzero, the event $(\alpha, \beta, \gamma) = (\alpha_0, \beta_0, \gamma_1)$ could be happen {\em almost never}, and we would not achieve the desired contradiction.

\section{Tools for the study of network-local sets}\label{App:Semialgebraic}

We review first the tools used to study the sets of Bell-local correlations, before going to the general case.
We obtain the following results in any network:
\begin{itemizeminus}
\item The membership problem (is $\vec{P} \in \mathcal{L}$?) can be solved, in principle, with tools of real algebraic geometry.
\item The network-local set of correlations $\mathcal{L}$ is closed, connected, and described by a finite number of nonstrict polynomial inequalities.
\end{itemizeminus}
Computational requirements, however, preclude the use of such tools except in the simplest cases.

\subsection{Bell-local sets}

All Bell-local correlations can be written as a convex mixture of deterministic behaviors. We follow the notation of the review~{\cite{Brunner2014}}:
\begin{equation}
  \vec{P} = \sum_{\lambda} q_{\lambda} \vec{d}_{\lambda} = D \vec{q}, \qquad
  \vec{q} \geqslant 0, \qquad \sum_{\lambda} q_{\lambda} = 1,
\end{equation}
where the inequality $\vec{q} \geqslant 0$ is understood componentwise.
We also collected the deterministic behaviors column-wise in the matrix $D$.
We write $\mathcal{L}$ the set of all Bell-local $\vec{P}$.
Note that the construction of the Bell-local set from deterministic strategies preserves a large part of the mathematical structure, as the labels $\lambda$ in the sum above contains the deterministic outputs of each input of each party.

\paragraph*{Solving the membership problem. ---}
We consider the following expression for vectors $\vec{\xi} \in \mathbbm{R}^{\dim (\vec{P})}$, $0 \leqslant \vec{\zeta} \in \mathbbm{R}^{\dim (\vec{q})}$:
\begin{equation}
  \label{Eq:ExprBell} I = \vec{\xi}^{\top} (\vec{P} - D \vec{q}) + \vec{\zeta}^{\top}  \vec{q} \geqslant 0
\end{equation}
which is nonnegative by construction (we put no term for the normalization of $\vec{q}$ as it follows from the normalization of $\vec{P}$).
Whenever $\vec{\zeta}^{\top} - \vec{\xi}^{\top} D = 0$, the inequality reduces to an inequality valid for all Bell-local correlations: $I = \vec{\xi}^{\top} \vec{P} \geqslant 0$.
Thus, when faced with a normalized $\vec{P}$, we can search for a certificate that proves $\vec{P} \notin \mathcal{L}$:
\begin{equation}
  \label{Eq:ProgramBell} \begin{array}{rll}
    \nu = \min_{\vec{\zeta}, \vec{\xi}}  \vec{\xi}^{\top} \vec{P} &  & \text{
    such that}\\
    \vec{\zeta} & \geqslant & 0\\
    \vec{\zeta}^{\top} - \vec{\xi}^{\top} D & = & 0.
  \end{array}
\end{equation}
which is a linear program.
When the optimal solution has $\nu^{\ast} < 0$, the behavior $\vec{P}$ is not local. 
Following the theory of polytope duality, all $\vec{P} \notin \mathcal{L}$ can be detected.

\paragraph*{Computing a complete description. ---}
We just showed how to solve the membership problem for Bell-local sets.
However, if we want a complete description of $\mathcal{L}$, we can obtain it by variable elimination.
Consider the set $X$:
\begin{equation}
  X = \left\{ \vec{x} = (\vec{P}, \vec{q}) \text{ s.t. } \vec{P} = D
  \vec{q}, \quad \vec{q} \geqslant 0, \quad \sum_{\lambda} q_{\lambda} = 1
  \right\} .
\end{equation}
This set is a polytope in $\mathbbm{R}^{\dim (\vec{P}) + \dim (\vec{q})}$ described by linear equalities and inequalities.
We now consider the projection of $X$ on the dimensions corresponding to $\vec{P}$:
\begin{equation}
  \label{Eq:ProjectionOnP}
  X_{| \vec{P}} = \left\{ \vec{P} \text{ such that } \exists \vec{q}, (\vec{P}, \vec{q}) \in X \right\} .
\end{equation}
The resulting polytope, equivalent to the local set $\mathcal{L}= X_{|\vec{P}}$ can be computed using Fourier-Motzkin elimination~{\cite{Williams1986}}.
The process is however quite demanding and has only been completed for simple scenarios.

\subsection{General networks}

In more general networks, the network-local correlation set $\mathcal{L}$ is no longer a polytope.
The reduction to deterministic strategies still works, as any source of randomness can be embedded in a local hidden variable~{\cite{Brunner2014}}.
However, the parties no longer have access to the same local hidden variable to perform a convex mixture of those strategies.

In some networks, such as those of bilocal scenarios~{\cite{Branciard2012}}, the reduction to finite models is done while preserving much of the structure.
As we will see in Appendix~\ref{App:SOS}, this can be helpful to reduce the complexity of the local model. Such constructions also work in other networks without loops using a construction similar to the one in~{\cite{Rosset2015a}}.

In networks involving loops, such as the triangle network of Figure~\ref{Fig:ThreeScenarios}a, we have to resort to Proposition~\ref{Prop:BetterBound}.
After computing the cardinality of all value sets, we parameterize the model using discrete probability distributions.
For example, for the triangle scenario, we need $3 \cdot (6 - 1)$ coefficients for the distributions of $\alpha, \beta, \gamma$ and $3 \cdot (2 - 1) \cdot 6^2$ coefficients to parameterize the local response functions (having removed the degree of freedom of the normalization), as in Eq.~(\ref{Eq:FiniteModelTriangle}).
We collect all these coefficients in a vector $\vec{q}$, along with linear inequalities that enforce their nonnegativity:
\begin{equation}
  g_j = \vec{a}_j^{\top}  \vec{q} - b_j \geqslant 0,
\end{equation}
where $\vec{a}_j \in \mathbbm{R}^{\dim (\vec{q})}, b_j \in \mathbbm{R}$.
The fact that $\vec{q}$ reproduces the behavior $\vec{P}$ is given by the relation~(\ref{Eq:GenericModel}), which we abbreviate
\begin{equation}
  f_i = \vec{c}_i^{\top}  \vec{P} - h_i = 0,
\end{equation}
where $\vec{c}_i \in \mathbbm{R}^{\dim (\vec{P})}$ and $h_i (\vec{q})$ is a polynomial in the coefficients of $\vec{q}$.

\paragraph*{Solving the membership problem. ---}
For all polynomials $F_i(\vec{P}, \vec{q})$, $G_{j k} (\vec{P}, \vec{q})$ and $L_l (\vec{P}, \vec{q})$, the following expression is nonnegative:
\begin{equation}
  \label{Eq:ExprGeneral} I = \sum_i F_i f_i + \sum_{j k} (G_{j k})^2 g_j +
  \sum_l (L_l)^2 \geqslant 0.
\end{equation}
In comparison with~(\ref{Eq:ExprBell}), we now have polynomial coefficients, and we employed squared polynomials to enforce nonnegativity\footnote{Being a sum-of-squares is a sufficient condition for a polynomial to be nonnegative.}.
When all the monomials in $\vec{q}$ cancel, we are left with a polynomial inequality in $\vec{P}$ valid for all network-local behaviors.
As with the linear program~(\ref{Eq:ProgramBell}), we can detect points outside $\mathcal{L}$ by the minimization:
\begin{equation}
  \begin{array}{rll}
    \displaystyle
    \nu = \min_{\{ F_i \}, \{ G_{j k} \}, \{ L_l \}} I (\vec{P}) &  & \text{
    such that}\\
    I (\vec{P}, \vec{q}) & = & I (\vec{P})
  \end{array}
\end{equation}
which can be formulated as a semidefinite program by upper bounding the degree of all $\{ F_i \}$, $\{ G_{j k} \}$, $\{ L_l \}$.
By increasing this bound step by step, we get a semidefinite hierarchy, proven to converge~{\cite{Lasserre2001,Parrilo2003}}.

In the triangle network, what are the requirements of this hierarchy for the bound of $| \Omega_{\alpha} | = | \Omega_{\beta} | = | \Omega_{\gamma} | \leqslant 6$ proven in Proposition~\ref{Prop:BetterBound}~?
First, observe that Eq.~(\ref{Eq:FiniteModelTriangle}) can be written:
\begin{equation}
  \label{eq:simplertriangle} P (a b c) = \sum_{\alpha, \beta, \gamma = 1}^6 P
  (a \beta | \gamma) P (b \gamma | \alpha) P (c \alpha | \beta),
\end{equation}
with $P (a \beta | \gamma) = P (\beta) P (a | \beta \gamma), \ldots$ and the constraints suitably modified.
The number of degrees of freedom of $P (a \beta | \gamma)$ is at least $r^2 + r - 1$, where $r = 6$ is the rank of the local hidden variables, with the total number of degrees of freedom $\mathcal{D}= 3 (r^2 + r - 1) = 123$.
The highest degree of involved polynomials is $3$ in Eq.~(\ref{eq:simplertriangle}).
Thus the relaxation should be of degree at least 2, and involve semidefinite matrices of row and column size $\mathcal{D} (\mathcal{D}+ 1) / 2 = 7626$, which is out of our reach.
However, the complexity decreases rapidly with $r$: for $r = 5$, we obtain the size $3828 \times 3828$, while $r = 4$ has size $1653 \times 1653$.
Thus, better upper bounds are not only interesting in theory; they render the membership problem tractable in practice.

\paragraph*{Computing a complete description. ---}
We can also characterize the set $\mathcal{L}$ by considering:
\begin{equation}
  X = \left\{ \vec{x} = (\vec{P}, \vec{q}) \text{ s.t. } f_i (\vec{P},
  \vec{q}) = 0 \text{ and } g_j (\vec{q}) \geqslant 0 \right\},
\end{equation}
which a subset of $\mathbbm{R}^{\dim (\vec{P}) + \dim (\vec{q})}$ characterized by polynomial (in)equalities: thus a semialgebraic set~{\cite{Bochnak1998}}. 
Its projection $X_{| \vec{P}}$ on the variables $\vec{P}$ is also a semialgebraic set. 
As $X$ is a compact set, by the Tube Lemma~{\cite{Rotman1998}}, $X_{| \vec{P}}$ is bounded and closed as well.
Thus $\mathcal{L}$ is a semialgebraic closed set. 
Using the Finiteness Theorem~{\cite[Thm 2.7.1]{Bochnak1998}}, $\mathcal{L}$ can be written as a finite union:
\begin{eqnarray}
  \label{Eq:FormLocalSet}
  \mathcal{L} &= & X_{| \vec{P}} = \bigcup_i \mathcal{L}_i, \\
  \mathcal{L}_i &=& \left\{ \vec{P} \in \mathbbm{R}^{\dim
  (\vec{P})} \text{ s.t. } f_{i j} (\vec{P}) \geqslant 0, \forall j \in
  J_i \right\}, \nonumber
\end{eqnarray}
and $| J_i | \leqslant A (A + 1) / 2$, where $A$ is the affine dimension of the nonsignaling space of behaviors $\vec{P}$, and the bound on the number of inequalities is given by the Br{\"o}cker-Scheiderer Theorem~{\cite[Thm 10.4.8]{Bochnak1998}}.
We also know that $\mathcal{L}$ is connected: by including a noise parameter in each local hidden variable, as done in the proof of Proposition~\ref{Prop:BetterBound}, it is possible to connect all behaviors to the uniformly random distribution.

In principle, $X_{| \vec{P}}$ can be computed using the cylindrical algebraic decomposition algorithm~{\cite{Basu2006}}, such as implemented by the \textverbatim{Reduce} function of Mathematica.
In practice, the process is extremely demanding and can only be completed in very simple cases~\cite{Lee2015}; moreover the output seldom matches the nice form~(\ref{Eq:FormLocalSet}).

In the next Appendix, we demonstrate the usefulness of sum-of-squares in the context of networks, by proving algebraically a numerical result provided in~{\cite{Branciard2012}}.

\section{Example of a sum-of-squares proof}\label{App:SOS}

\begin{table*}[t]
  $\begin{array}{|l|l|l|l|}
    \hline
    = \langle .. \rangle & 1 & C_0 & C_1\\
    \hline
    1 &  & 1 - \eta & 1 - \eta\\
    \hline
    A_0 & \eta - 1 & - (\eta - 1)^2 & - (\eta - 1)^2\\
    \hline
    A_1 & 1 - \eta & (\eta - 1)^2 & (\eta - 1)^2\\
    \hline
  \end{array} \quad \begin{array}{|l|l|l|l|}
    \hline
    = \langle .B_0 . \rangle & 1 & C_0 & C_1\\
    \hline
    1 & 0 & 0 & 0\\
    \hline
    A_0 & 0 & \eta^2 / 2 & - \eta^2 / 2\\
    \hline
    A_1 & 0 & - \eta^2 / 2 & \eta^2 / 2\\
    \hline
  \end{array}$
  
  \
  
  \
  
  $\begin{array}{|l|l|l|l|}
    \hline
    = \langle .B_1 . \rangle & 1 & C_0 & C_1\\
    \hline
    1 & 0 & 0 & 0\\
    \hline
    A_0 & 0 & \eta^2 / 2 & \eta^2 / 2\\
    \hline
    A_1 & 0 & \eta^2 / 2 & \eta^2 / 2\\
    \hline
  \end{array} \quad \begin{array}{|l|l|l|l|}
    \hline
    = \langle .B_0 B_1 . \rangle & 1 & C_0 & C_1\\
    \hline
    1 & 0 & 0 & 0\\
    \hline
    A_0 & 0 & 0 & 0\\
    \hline
    A_1 & 0 & 0 & 0\\
    \hline
  \end{array}$
  \caption{\label{Table:nloc:2loc}All correlators involved in a $2$-locality
  test involving two singlet states and inefficient detectors for Alice and
  Charlie, with efficiency $\eta$.}
\end{table*}

\begin{table*}
  $\begin{array}{|l|l|l|l|l|}
    \hline
    = \langle .. \rangle & 1 & C_0 & C_1 & C_0 C_1\\
    \hline
    1 &  & 1 - \eta & 1 - \eta & \xi\\
    \hline
    A_0 & \eta - 1 & - (\eta - 1)^2 & - (\eta - 1)^2 & (\eta - 1) \xi\\
    \hline
    A_1 & 1 - \eta & (\eta - 1)^2 & (\eta - 1)^2 & (1 - \eta) \xi\\
    \hline
    A_0 A_1 & \zeta & (1 - \eta) \zeta & (1 - \eta) \zeta & \xi \zeta\\
    \hline
  \end{array} \quad \begin{array}{|l|l|l|l|l|}
    \hline
    = \langle .B_0 . \rangle & 1 & C_0 & C_1 & C_0 C_1\\
    \hline
    1 & 0 & 0 & 0 & 0\\
    \hline
    A_0 & 0 & \eta^2 / 2 & - \eta^2 / 2 & 0\\
    \hline
    A_1 & 0 & - \eta^2 / 2 & \eta^2 / 2 & 0\\
    \hline
    A_0 A_1 & 0 & f_2 & - f_2 & 0\\
    \hline
  \end{array}$
  
  \
  
  \
  
  $\begin{array}{|l|l|l|l|l|}
    \hline
    = \langle .B_1 . \rangle & 1 & C_0 & C_1 & C_0 C_1\\
    \hline
    1 & 0 & 0 & 0 & 0\\
    \hline
    A_0 & 0 & \eta^2 / 2 & \eta^2 / 2 & f_1\\
    \hline
    A_1 & 0 & \eta^2 / 2 & \eta^2 / 2 & f_1\\
    \hline
    A_0 A_1 & 0 & 0 & 0 & 0\\
    \hline
  \end{array} \hspace{7em} \begin{array}{|l|l|l|l|l|}
    \hline
    = \langle .B_0 B_1 . \rangle & 1 & C_0 & C_1 & C_0 C_1\\
    \hline
    1 & 0 & 0 & 0 & 0\\
    \hline
    A_0 & 0 & 0 & 0 & 0\\
    \hline
    A_1 & 0 & 0 & 0 & 0\\
    \hline
    A_0 A_1 & 0 & 0 & 0 & 0\\
    \hline
  \end{array}$
  \caption{
\label{Table:nloc:model}
Bilocal model for the correlations after symmetrization. The number of variables has been greatly reduced.
}
\end{table*}

We give below an example of sum-of-squares relaxations applied to the
characterization of network-local sets, taken from the section III.C.2
of~{\cite{Branciard2012}}. There, the authors studied the correlations coming
from an entanglement-swapping experiment and its resistance to detector
inefficiencies. However, the result $\eta_{\tmop{biloc}} = 2 / 3$ was obtained
numerically. We provide an algebraic proof below using a sum-of-squares
decomposition.

\subsection{$2$-local model}

First, we summarize the $2$-local model considered in~{\cite{Branciard2012}}.
We consider an entanglement swapping experiment, where the source
$\mathcal{S}_1$ is connected to Alice and Bob, while the source
$\mathcal{S}_2$ is connected to Bob and Charlie. Bob does not receive an
input, but outputs two bits, $B_0 = \pm 1$ and $B_1 = \pm 1$. Alice and
Charlie have binary inputs $x, y = 0, 1$ with corresponding outputs $A = \pm
1$ and $C = \pm 1$. From the bilocality condition, we write:
\begin{multline}
  \label{eq:bilocalconstraint} P (A B_0 B_1 C | x z) =\\ \int P (A | x
  \lambda_1) P (B_0 B_1 | \lambda_1 \lambda_2) P (C | z \lambda_2) \mathd
  \pi_1 (\lambda_1) \mathd \pi_2 (\lambda_2) .
\end{multline}
In~{\cite{Branciard2012}}, this was shown to be equivalent to the existence of
an underlying distribution of $A_0$, $A_1$, $B_0$, $B_1$, $C_0$, $C_1$, where
$A_x$ and $C_z$ are the deterministic values of $A$ and $C$ for all inputs. As
these are $\pm 1$-valued variables, we can compute the averages $\langle
A_0^{i_0} A_1^{i_1} B_0^{j_0} B_1^{j_1} C_0^{k_0} C_1^{k_1} \rangle$ for $i_0,
i_1, j_0, j_1, k_0, k_1 = 0, 1$. As $A_0, A_1$ and $C_0, C_1$ can all be
observed at the same time, the following distribution exists and is
nonnegative:
\begin{align}
  \label{eq:bilocalpos} P (A_0 = \alpha_0, A_1 = \alpha_1, B_0 = \beta_0,
    B_1 = \beta_1, C_0 = \gamma_0, C_1 = \gamma_1) & \nonumber \\
    = \frac{1}{64} \Big \langle (1 + \alpha_0 A_0) (1 + \alpha_1 A_1) (1 + \beta_0
    B_0) \qquad \qquad &  \nonumber \\
    (1 + \beta_1 B_1) (1 + \gamma_0 C_0) (1 + \gamma_1 C_1) \Big \rangle \geqslant 0, &
\end{align}
and the independence of $\Lambda_1$, $\Lambda_2$ reduces to:
\begin{equation}
  \langle A_0^{i_0} A_1^{i_1} C_0^{k_0} C_1 \rangle = \langle A_0^{i_0}
  A_1^{i_1} \rangle \langle C_0^{k_0} C_1^{k_1} \rangle
\end{equation}
for all $i_0, i_1, k_0, k_1 = 0, 1$. The $\langle \ldots \rangle$ containing
both $A_0 A_1$ or both $C_0 C_1$ cannot be obtained from the distribution $P
(A B_0 B_1 C | x z)$ and are unknown parameters of the model.

\subsection{Quantum correlations to test}

The quantum correlations are obtained from the state $\rho_{\text{ABC}} =
\rho_{\tmop{AB}} \otimes \rho_{\tmop{BC}} = | \Psi^- \rangle \langle \Psi^- |
\otimes | \Psi^- \rangle \langle \Psi^- |$ distributed by the sources
$\mathcal{S}_1$ and $\mathcal{S}_2$. Alice and Charlie have both inefficient
detectors with detection efficiency $\eta$. When $\eta = 1$, Alice and Charlie
use the projective measurements with $\pm 1$-valued outcomes:
\begin{equation}
  \overline{A}_0 = \overline{C}_0 = \frac{\sigma_z + \sigma_x}{\sqrt{2}},
  \qquad \overline{A}_1 = \overline{C}_1 = \frac{\sigma_z -
  \sigma_x}{\sqrt{2}},
\end{equation}
but in case of nondetection, Alice outputs $A = (- 1)^{x + 1}$ and Charlie
always outputs $C = + 1$. Thus, their effective measurement operators are:
\begin{align}
  A_0 &= \eta \overline{A}_0 - (1 - \eta) \mathbbm{1}, & A_1 &= \eta
  \overline{A}_1 + (1 - \eta) \mathbbm{1},  \nonumber \\
  C_z &= \eta \overline{C}_z + (1 - \eta) \mathbbm{1} . & &
\end{align}
Bob performs a Bell state measurement $\{ \mathcal{B}_{B_0 B_1} \}$ with four
outcomes, described by two bits $B_0, B_1 = \pm 1$:
\begin{align}
  \mathcal{B}_{+ +} &= | \Phi^+ \rangle \langle \Phi^+ |, & \mathcal{B}_{- +} & = | \Phi^- \rangle \langle \Phi^- |, \nonumber \\
  \mathcal{B}_{+ -} &= | \Psi^+ \rangle \langle \Psi^+ |, & \mathcal{B}_{- -} &= | \Psi^- \rangle \langle \Psi^- |,
\end{align}
so that the value of the bits $B_0$, $B_1$ and their parity $B_0 B_1$, are
given by:
\begin{align}
  B_0 &=\mathcal{B}_{+ +} +\mathcal{B}_{+ -} -\mathcal{B}_{- +} -\mathcal{B}_{- -}, \nonumber \\
  B_1 &=\mathcal{B}_{+ +} -\mathcal{B}_{+ -} +\mathcal{B}_{- +} -\mathcal{B}_{- -}, \nonumber \\
  B_0 B_1 &=\mathcal{B}_{+ +} -\mathcal{B}_{+ -} -\mathcal{B}_{- +} +\mathcal{B}_{- -} .
\end{align}
The corresponding correlations are given in Table~\ref{Table:nloc:2loc}.

\subsection{Simplified $2$-local model}

We find that the correlations are invariant under a symmetry group $G$ of
order $8$, with generators $g_{\tmop{AB}}$, $g_{\tmop{BC}}$ and
$g_{\tmop{ABC}}$:
\begin{align}
  \begin{array}{rl}
    g_{\text{AB}} & =
                    \left\{ \begin{array}{lr}
    A_0 \rightarrow & - A_1\\
    A_1 \rightarrow & - A_0\\
    B_1 \rightarrow & - B_1
  \end{array} \right.,
    \\
    & \\
    g_{\text{BC}} & =
                    \left\{ \begin{array}{lr}
    C_0 \rightarrow & C_1\\
    C_1 \rightarrow & C_0\\
    B_0 \rightarrow & - B_0
   \end{array} \right.,
   \end{array} &&
  g_{\text{ABC}} &= \left\{ \begin{array}{lr}
    A_0 \rightarrow & - C_0\\
    A_1 \rightarrow & C_1\\
    C_0 \rightarrow & - A_0\\
    C_1 \rightarrow & A_1\\
    B_0 \rightarrow & B_1\\
    B_1 \rightarrow & B_0
  \end{array} \right. . 
\end{align}
Both $g_{\tmop{AB}}$ and $g_{\tmop{BC}}$ can be applied on any bilocal model
by embedding a flag in $\mathcal{S}_1$ and $\mathcal{S}_2$ respectively, while
$g_{\tmop{ABC}}$ is not compatible with the model. We thus apply the
symmetries $g_{\tmop{AB}}$ and $g_{\tmop{BC}}$ on our bilocal model filled
with the values of Table~\ref{Table:nloc:2loc} to obtain the correlations in
Table~\ref{Table:nloc:model}, written using four unknowns $\xi$, $\zeta$,
$f_1$, $f_2$ and the detection efficiency $\eta$. We thus need to find the
range of $\eta$ such that values for the unknown variables exist and
Eq.~(\ref{eq:bilocalpos}) is satisfied (the other constraints are already
taken care of).

The solution ($\eta \in [0, 2 / 3]$) was found using a sum-of-squares
relaxation with the solver SOSTOOLS~{\cite{sostools}}; we then worked the
analytic solution below from the numerical output of the solver.

\subsection{Sum-of-squares solution}

Among the inequalities coming from~(\ref{eq:bilocalpos}), the following
inequalities are relevant, writing $\overline{\xi} = 1 - \xi$ and
$\overline{\zeta} = 1 + \zeta$:
\begin{equation}
  \begin{array}{rlll}
    0 \leqslant g_1 = & 64 P (+ + + + - +) & = & \overline{\xi \zeta} - 2 (f_1 +
    f_2) ,\\
    0 \leqslant g_2 = & 64 P (- + - + - +) & & \\
    = & \multicolumn{3}{l}{4 \overline{\xi} - \xi \zeta - 2 \eta^2 - 2 \overline{\xi} \eta - 2 f_2 ,} \\
    0 \leqslant g_3 = & 64 P (+ - + + - -) & = & \overline{\xi \zeta} - 2 \eta
    \overline{\zeta} - 2 \overline{\xi} \eta + 4 \eta^2,\\
    0 \leqslant g_4 = & 64 P (- + + + - -) & & \\
    = & \multicolumn{3}{l}{\overline{\xi \zeta} +
    \overline{\xi} (2 \eta - 4) - \overline{\zeta} 2 \eta + 8 \eta - 2 \eta^2
    ,}\\
    0 \leqslant g_5 = & 64 P (+ - + + - +) & = & 2 \overline{\xi} \eta -
    \overline{\xi \zeta} - 2 \eta^2 + 2 f_2,\\
    0 \leqslant g_6 = & 64 P (+ + + + - -) & = & 2 \overline{\zeta} \eta -
    \overline{\xi \zeta} - 2 \eta^2 + 2 f_1.
  \end{array}
\end{equation}

We form the conic combinations:
\begin{align}
    0 \leqslant & F_+ = g_1 + 2 g_3 + g_5 + g_6 = \left( 2 \eta - \overline{\zeta}
    \right) \left( 2 \eta - \overline{\xi} \right), \nonumber \\
    0 \leqslant & F_-  = 2 g_1 + g_3 + 2 g_5 + 2 g_6 = - \left( 2 \eta -
    \overline{\zeta} \right) \left( 2 \eta - \overline{\xi} \right), \nonumber \\
    0 \leqslant & I = (2 g_1 + g_2 + 3 g_5 + 2 g_6) = \overline{\xi} + \eta
    \overline{\xi} + \eta \overline{\zeta} - \overline{\xi \zeta} - 3 \eta^2 , \nonumber \\
    0 \leqslant & J = (3 g_3 + g_4) / 4 =  \left( \eta - \overline{\zeta} + 1
    \right) \left( 2 \eta - \overline{\xi} \right),
\end{align}
From $F_{\pm}$, we deduce that either $\left( 2 \eta - \overline{\zeta}
\right) = 0$ or $\left( 2 \eta - \overline{\xi} \right) = 0$. We examine first
$\overline{\zeta} = 2 \eta$, and substitute:
\begin{equation}
  I = \overline{\xi} (1 - \eta) - \eta^2 \geqslant 0, \quad J = (1 - \eta)
  \left( 2 \eta - \overline{\xi} \right) \geqslant 0.
\end{equation}
\begin{definition}
  A polynomial $p$ to be a {\em sum-of-squares polynomial} (SOS) if there
  is a decomposition:
  \begin{equation}
    p = \sum_k (m_k)^2,
  \end{equation}
  where $m_k$ are polynomials in $\xi, \eta$.
\end{definition}

For SOS polynomials $p_1, p_2$ and $p_3$, we have $K = p_1 + p_2 I + p_3 J
\geqslant 0$. To find a bound on the detection efficiency
$\eta_{\text{bound}}$, we write:
\begin{equation}
  K = \eta_{\text{bound}} - \eta = p_1 + p_2 I + p_3 J \geqslant 0,
\end{equation}
which will prove that bilocal models satisfy $\eta \leqslant
\eta_{\text{bound}}$. The bound $\eta_{\text{bound}} = 2 / 3$ is verified
using:
\begin{equation}
  \label{Eq:nloc:2locdetineq} p_1 = \frac{(2 - 3 \eta)^2}{6}, \quad p_2 = p_3
  = \frac{1}{2} .
\end{equation}
The procedure for $\overline{\xi} = 2 \eta$ is similar and also gives
$\eta_{\tmop{bound}} = 2 / 3$.

\end{document}